\documentclass[a4paper,12pt]{article}%
\usepackage{amssymb}
\usepackage{amsfonts}
\usepackage{amsmath,enumerate}
\usepackage{graphicx}
\usepackage[english]{babel}
\usepackage{amssymb}
\usepackage{amsfonts}
\usepackage{amsmath}%
\providecommand{\U}[1]{\protect\rule{.1in}{.1in}}

\newtheorem{propo}{Proposition}

\newtheorem{remark}{Remark}
\newtheorem{coro}{Corollary}
\newenvironment{proof}[1][Proof]{\noindent\textbf{#1.} }{\ \rule{0.5em}{0.5em}}

\def\la{\lambda}

\def\siia{\Updownarrow}

\begin{document}

\title{The Home Office in Times of COVID- 19 Pandemic and its impact in the  Labor Supply}
\author{Jos\'e Nilmar  Alves de Oliveira \\Processamento de Dados Amazonas S.A., \\Governo do Estado do Amazonas, Brasil
\and Jaime Orrillo\\Catholic University of Brasilia, Brazil
\and Franklin Gamboa\\Federal University of Mato Grosso, Brazil}
\maketitle

\begin{abstract}
We lightly modify Eriksson's (1996) model to accommodate the home office in a simple model of endogenous growth. By home office we mean any working activity carried out away from the workplace which is assumed to be fixed. Due to the strong mobility restrictions imposed on citizens during the COVID-19 pandemic, we allow the home office to be located at home. At the home office, however, in consequence of the fear and anxiety workers feel because of COVID-19, they become distracted and spend less time working. 

We show that in the long run, the intertemporal elasticity of substitution of the home- office labor is sufficiently small only if the intertemporal elasticity of substitution of the time spent on distracting activities is small enough also.

\hfill \break \vspace{ .5 mm }\break {\bf Keywords: } COVID-19;  Social Distancing; Home Office; Economic growth model.\\
{\bf JEL Codes}: O41.
\end{abstract}

\newpage

\section{Introduction}
The home office is  a modality of working  away from  a fixed job location
the office and therefore  can be  carried  out in any place different from a physical office. This modality of labor  has  existed  for a long time   and has been  mainly common in multinational enterprises\footnote{ Where some   workers like   managers needed   to get away from the office.}.   However, this modality of working  has increased over time around the world as shown in a survey carried out by Ipsos, see Fig. 1 below.

\begin{figure}[h]
\includegraphics[width=0.8\textwidth]{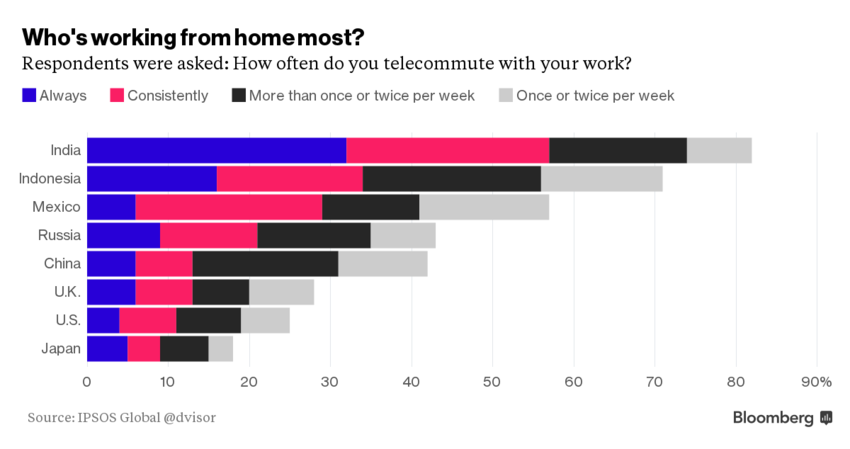}
\caption{ India, Indonesia and Mexico had the highest amount of telecommuters.}
\label{figura}
\end{figure}

In spite of the  significant gains from working from home in terms of worker productivity and satisfaction as shown in Bloom et al. (2015), it does not seem to be the case during   the COVID-19 pandemic, as  workers could engage in distracting activities\footnote{Due to the fear and the  anxiety  produced by the current pandemic.}   placing workers' productivity at risk. In this same vein of reasoning,    Dutcher and Jabs Saral (2012) highlight, even in normal times,  the difficulties that may arise if telecommuting workers are not properly monitored.

In spite of the fact that the home office or working from home\footnote{Working from  home  is  a very special kind of home office and it seems to be the rule in times of the current pandemic.} in the past it was only applied to specific  jobs,   nowadays,  because   the COVID-19 pandemic,  most firms seem to be obligated to adopt it  to  continue operating in their markets. Consequently, both the demand and supply for home offices seem to increase. Hence, it is the purpose  of this paper to  theoretically determine  the factors which influence  the  home-office job supply.

Of course not all jobs can be accomplished in a remote way \footnote{The  construction sector is one from them for instance.}  This makes the productive sector seek to reorganise itself so that working from home becomes the best alternative due to the strong mobility restrictions imposed by the authorities of countries.    It is also well known that many enterprises are planning on adopting home offices even after  the current pandemic. So much so, that in many countries some businesses, like hotels, are planning to adapt their spaces to offer them as places to set up home offices. This obliges us to make a long-run analysis of the home-office job supply.

To accomplish our purpose we consider a simple economic growth model with an endogenous labor supply like that of Eriksson (1996). Although our model is very similar to Eriksson?s model, we depart from it in the following aspect. We assume that the effort attached to human capital depends on the time spent on distracting activities, occurring during the working period. We assume that these distracting activities give some pleasure to workers, but we should not confuse such activities with leisure activities since the latter are supposed to occur outside normal working hours.

It is also useful to note that during the COVID-19 pandemic, workers are put on quarantine. It is therefore natural to assume that workers are more prone to be distracted by activities that decrease their home office production. We can interpret such distracting activities as being negative ``shocks'' to the labor supply. The word shock used here is not used in its rigorous sense of being stochastic since the economic growth model we used here is deterministic\footnote{ However, in our model,  the time  spent  on distracting activities is supposed to  change over time, reflecting the  intrinsic uncertainty imposed by the quarantine which makes workers choose, in a random way,  those  activities which result  in a decline of hours worked. }.

Our paper is related to the recent literature on the economic effects of the COVID-19 crisis. The papers in this literature have different objectives - from understanding its evolution to predicting its impacts on the world economy. Nonetheless, our paper is more related to the working-from-home literature which is surveyed by Allen et al. (2015). These authors address a type of home office, namely telecommuting and analyse how effective it is.

 Our objective in this paper is much more modest in the sense that we seek to discover how the preferences for consumption and displeasure for working affect the growth of home- office job supply. More precisely, we show that this growth rate is affected by the parameters that represent both the workers? preferences for consumption and the workers? displeasure for working. This is important for both firms and governments as it allows them to implement, in an optimal way, the incentives for home-office job supply by adjusting the goals of each policy maker. Our specific findings are related to the growth rate of variables along balanced-growth solution paths. We show that in the long run the intertemporal elasticity of substitution of home-office labor is sufficiently small only if the intertemporal elasticity of substitution of the time spent on distracting activities is small enough too.

The paper is organised as follows: Section 2 presents the model in which are described both households and firms. In this same section, the social planner?s problem is formulated. Section 3 presents the main results and the paper ends with a short section which analyses the theoretical results and gives some concluding remarks.

\section{The model} 

Our model is a centralised economy like that of Eriksson (1996) with some key modifications allowing us to make the labor supply endogenous via distracting activities. By distracting activities here I mean any activity which decreases labor time. They are not properly leisure but rather activities which produce both pleasure and affect or influence the acquisition of human capital which is placed in motion to produce consumption good/capital.

\begin{remark}
Distracting activities influence the effort which is allocated to produce either consumption goods or capital for future investments. Moreover, these activities also influence the effort used to acquire human capital.
\end{remark}

\subsection{Households} 
 We allow each worker, within  a constant population,  to supply labor $l_t.$    We assume that during  the working hours    each worker spends  time  $s_t$ in distracting activities  which give a certain pleasure.  It is useful to pointing  out that  the time $s_t$ is not leisure since it occurs during  working hours  $l_t.$ Thus, $l_s - s_t$ is the effective labor.

We model the instantaneous payoff of each worker to be 
$$V(c_t, s_t, l_t) = u(c_t) - v(l_t - s_t)\eqno(1)$$
This payoff consists of two parts: the former is the pleasure coming from consumption and the latter is the displeasure  felt while  working.

\subsection{Production} 
On the production side of the economy,     the output   is produced by using both capital, $k_t$  and effective labor $l_t-s_t$ which is potentiated by  the human capital  $h_t.$   We formalise this by assuming the effort depends on $s_t.$

Differently from   Uzawa  (1965) and Lucas (1988), we make  human capital endogenous by postulating  each  worker's effort  $e:[0, l_t]\to [0,1]$   depends on the time spent $s_t$ in the following way: $e(s_t) = 1 - \frac{s_t}{l_t}.$  We also  assume that  this effort is allocated both to the  production of  human capital and the consumption good/capital. Our allocation is  $(1 - e)$ for the former and $e$ for the latter.   More precisely, we assume that both  human and physical (or real) capitals  accumulate  according to 

$$\dot h_t = h_t( 1 - e(s_t))\eqno(1)$$
$$\dot k_t =  f(k_t, h_t e(s_t)(l_t-s_t))- c_t\eqno(2)$$ 

respectively.    Using the functional form of $e,$ we can rewrite  (1) and (2)  to be 
$$\dot h_t = h_t\frac{s_t}{l_t} \eqno(3)$$
$$\dot k_t =  f(k_t, \frac{h_t}{l_t}  (l_t-s_t)^2)- c_t\eqno(4)$$

\subsection{The planner}
After describing the consumption and production sides, we are going to formulate  the benevolent planner's problem.  The benevolent planner's problem is then  to chose  paths of consumption $c,$ labor $l,$ and   time spent $s$ on distracting activities    in order to  maximise the discounted stream  of payoffs by every identical agent in the economy 
$$\int_{0}^{\infty} e^{-\rho t}V(c_t,l_t, s_t)dt\eqno(5)$$
subject to   (3) and (4) with $k_o > 0$  and $h_o >$ given. 

Here, $c_t,l_t$ and $s_t$ are control variables and $\rho$ is the discount factor; and $k_t$ and $h_t$ are state variables.  In what follows we will drop all  indexes of time from variables which depend on the time  to attain analytical tractability.

In this model  the effort $e$ is endogenous,  not by itself, like  Eriksson (1996),  but  because  it depends on the time spent on distracting activities. However,   for the sake of comparison with related literature and mainly with that of   Eriksson's (1996) paper, we maintain,  like him,   that  the instantaneous  payoff  is  additively separable: 
$$V(c,s,l) = \frac{c^{1-\sigma}}{1-\sigma} - \frac{ (l-s)^{1+ \gamma}}{1 + \gamma},  1\not = \sigma > 0, \gamma >0,\eqno(6)$$
and the production function like a  Cobb-Douglas one :
$$f(k, he(l-s)) =  k^{\beta} [\frac{h}{l} (l-s)^2]^{1 - \beta}.\eqno(7)$$

\section{Theoretical results}
We begin this section by characterising the solutions of the benevolent planner?s problem.  For that, we consider  the current-value  Hamiltonian for the optimal problem, with ``prices''  $\lambda_1$  and $\lambda_2$  used to value increments to physical and human capital    respectively. 
$$H = \frac{c^{1-\sigma}}{1-\sigma} - \frac{ (l-s)^{1+ \gamma}}{1 + \gamma} + \lambda_1 [k^{\beta} [\frac{h}{l} (l-s)^2]^{1 - \beta}- c]  + \lambda_2 h\frac{s}{l}.$$

The  necessary  conditions for optimality  are:

\begin{enumerate}
\item $$H_c = 0  \Leftrightarrow  c^{-\sigma} = \lambda_1\eqno(8)$$
\item $$H_s = 0 \Leftrightarrow   (l-s)^{\gamma}  + (1-\beta) \lambda_1 k^{\beta}  [\frac{h}{l} (l-s)^2]^{-\beta} h(-2 +\frac{2s}{l})  + \lambda_2 \frac{h}{l}=0\eqno(9)$$
\item $$H_{l} = 0  \Leftrightarrow  -(l-s)^{\gamma}  +   \la_1 (1-\beta)  k^{\beta}  [\frac{h}{l} (l-s)^2]^{-\beta} h ( 1 -\frac{s^2}{l^2}) -\la_2h\frac{s}{l^2}= 0 \eqno(10)$$

\item 

$$\dot \lambda_1 = \rho \lambda_1 - H_k$$
$$\siia$$
$$\dot \lambda_1 = \lambda_1\{ \rho -  \beta k^{\beta -1} [\frac{h}{l} (l-s)^2]^{1-\beta} \}.\eqno(11)$$

\item $$\dot \lambda_2 = \rho \lambda_2 - H_h $$
$$\siia$$
$$\dot \lambda_2 =  \lambda_2 \rho -  \lambda_1 (1-\beta) k^{\beta}  [\frac{h}{l} (l-s)^2]^{-\beta}\frac{(l-s)^2}{l}  - \lambda_2 \frac{s}{l}\eqno(12)$$
\end{enumerate} 

(8) -(10 ) are the the first order conditions  to  to maximise $H,$ and    (11) and (12) give the rates of change  of  $\lambda_i, i=1,2$  of both capitals.

\subsection{Steady state path}

Next, we will  seek the balanced growth path from  (8)-(12)   which are solutions on which consumption and both kinds of capital are growing at constant percentage rates, the prices of the two kinds of capital are declining at constant rates.

In what follows we will  present our first  results which have to do with  the relationships between the rates of change of our variables $k, c, l-s$ and $h.$

\begin{propo}
Along the balanced growth solution of the planner's problem, we have 
$$\hat k = \hat c = \hat h + 2 \hat {(l-s)}-\hat l$$

\end{propo}
\begin{proof}

From (8) we have $\hat \lambda_1 = - \sigma \hat c.$ So that if $\hat c = \theta,$ then from (11) we get that the marginal productivity of capital is constant. That is, 
$$ k^{\beta -1}  [\frac{h}{l} (l-s)^2]^{1-\beta}  = \frac{ \rho + \sigma \theta}{ \beta}.\eqno(13)$$

Dividing by $k$ through (4) and using (7), we have 
$$\frac{ c}{k} = k^{\beta -1} [\frac{h}{l} (l-s)^2]^{1-\beta} - \hat k \eqno(14)$$ 
Using (13)  and the fact that   $\hat k$ grows at  a constant rate,  by hypothesis,  one has $\frac{ c}{k} $ is a constant.
After Differentiating  (14) logarithmically with respect to time we get  
$$\hat k = \hat c = \theta.\eqno(15)$$

Differentiating  (13) logarithmically with respect to time  one has that the common growth rate of consumption and  capital      is 
 $$\theta = \hat h + 2 \hat {(l-s)}-\hat l\eqno(16)$$

\end{proof}

\begin{propo}
Under the same hypotheses of Proposition 1, one has:
\begin{enumerate}
\item $ \hat s = \hat l =\hat{ (l-s)}  = \frac{ 1 - \sigma}{ \gamma + \sigma} \hat h $
\item $\hat k = \hat  c =\frac{ 1 + \gamma}{\gamma + \sigma}\hat h$
\end{enumerate} 

\end{propo}

\begin{proof}

Summing (9) and (10) we  get 
$$(1-\beta) k^{\beta}  [\frac{h}{l} (l-s)^2]^{-\beta} \frac{\la_1}{\la_2} (l-s) = 1\eqno(17)$$

 Differentiating  (17) logarithmically with respect to time  and using (15) and (8) one has
$$\hat \la_2 = (1-\beta)\hat{(l-s)}  -\beta \hat h + (\beta - \sigma ) \theta\eqno(18)$$

Manipulating (9) and (17)  we have 
$$ (l-s)^{\gamma} = 2\frac{h\la_2}{l} \eqno(19)$$

  Differentiating  (19) logarithmically with respect to time  we have 
  $$\gamma  \hat{ (l-s)}  = \hat \la_2 + \hat h - \hat l \eqno(20)$$
  
 Manipulating (12) and using (17) we have $\hat{l-s} = \hat l.$ Using this fact,  (16) and  (20) become
  $$\theta = \hat h +  \hat {(l-s)} \eqno(21)$$
 $$(\gamma +1) \hat{ (l-s)}  = \hat \la_2 + \hat h  \eqno(22)$$
  respectively.

Putting (18) and (21) into (22) and  after arranging it  we have.

$$\hat{ (l-s)}  = \frac{ 1 - \sigma}{ \gamma + \sigma} \hat h \eqno(23)$$

Since $\hat h$ is constant we have that $\hat l = \hat s$ by (3).  Using the fact that $\hat{l-s} = \hat l,$  stablished above,   we get Item 1.   
Putting   (23) into (21) we get 
$$\theta = \hat h + \frac{ 1 - \sigma}{ \gamma + \sigma} \hat h = ( 1 +  \frac{ 1 - \sigma}{ \gamma + \sigma})\hat h\eqno(24)$$
Hence,  Item 2 follows. 

\end{proof}

The following proposition shows that  the growth rates in terms of parameters.

\begin{propo}
Under the same hypotheses of Proposition 1, one has:
\begin{enumerate}
\item $ \hat s = \hat l =\hat{ (l-s)}  =  \frac{(1 - \sigma)(\rho -1)}{1 - \sigma(2+\gamma)} $
\item $\hat k = \hat  c =\frac{(1 + \gamma)(\rho -1)}{1 - \sigma(2+\gamma)} $
\end{enumerate}

\end{propo}

\begin{proof}

After manipulating  (12) and to use (17) and (3) we  reach to 
 
$$\hat \la_2 = \rho -1 \eqno(25)$$

Putting (25) into (22) one has
$$(1+ \gamma) \hat{ (l-s)}  =  \rho -1  + \hat h\eqno(26)$$
Putting (21) into (26) we have 
$$(1+\gamma) \theta = (2+\gamma)\hat h +\rho -1  \eqno(27)$$

Putting (24) into (27) and after simplifying the result we get
$$\hat h = \frac{(\gamma + \sigma)(\rho -1)}{1 - \sigma(2+\gamma)}\eqno(28)$$

Finally, Proposition 3 follows after substituting (28) into Items 1 and 2 of Proposition 2. 

\end{proof}

The following corollary shows the  growth rate of the output equals the  growth rate of the capital (or consumption).  

\begin{coro}
Under assumptions  of Proposition 3, one has 
$$\hat y = \hat k = \hat c$$
\end{coro}
\begin{proof}
This result  follows from   Differentiating  (7) logarithmically with respect to time  and using (16) and  Item 2 of Proposition 3.    
\end{proof}

\subsection{Convergent utility and  transversality conditions} 
In  this section, we establish   the convergence of the utility function and    the transversality condition.  We will do it by considering the balanced path and  under the assumption that $\sigma > 1.$

First, the utility integral can be written as 
$$U = A_1 \int_{0}^{\infty} e^{  ( (1-\sigma)\hat c - \rho)t } dt -     A_2\int_{0}^{\infty} e^{ ( (1+ \gamma)\hat {(l-s)} - \rho )t } dt $$
where $A_1 = \frac{c_o^{1-\sigma}}{1- \sigma}$ and $A_1 = \frac{(l_o-s_o)^{1+\gamma}}{1+ \gamma}.$  Substituting the values of $\hat c$ and $\hat{l-s}$ given by Proposition 3 we have that $U$ is 
$$U = (A_1-A_2) \int_{0}^{\infty} e^{xt } dt  $$
where $x = \frac{(1+\gamma)( 1- \sigma) (\rho -1)  }{ 1 - \sigma(2 + \gamma)} - \rho$ is negative since $\sigma > 1$ and $\rho  < 1.$ Thus, $U$ is finite since $x$ is negative. 

Second, the transversality conditions  associated with  the  benevolent planner's problem  also hold. That is to say, 
$$\lim_{t\to \infty} k\la_1 e^{-\rho t} = 0 \hbox{ and } \lim_{t\to \infty} h\la_2e^{-\rho t} = 0$$
The former  follows from the  facts $\hat \la_1 = -\sigma \hat c$ and $\hat k = \hat c.$ Thus, one has 
$$k\la_1e^{-\rho t}  = k_o(\la_1)_o e^{( \hat k - \sigma \hat c)t } = k_o (\la_1)_0 e^{ [ (1-\sigma) \hat c -\rho] t}$$
Since $ (1-\sigma) \hat c -\rho = x$ and $x$ is negative from its definition (see above), the former result follows. 

Finally, the latter follows  from (25) and (28). To see it, it suffice to observe that

$$h\la_2 e^{-\rho t} = h_o(\la_2)_o e^{ (\hat h+ \hat \la_2 - \rho)t}$$

where  $ \hat h + \hat \la_1 -\rho$ equals $x.$  Since $x$  is negative,   the latter result follows.

\subsection{The intertemporal elasticity of  distracting activities} 
We start by setting the  elasticity  of the marginal utility  of time spent on   distracting activities, $V_2(c_t, s_t, l_t).$ Here  the sub-index  represents the  partial derivative with respect to $s_t.$   

By definition of elasticity one has 
$${\cal E}_V(s_t) =  \frac{\partial V_2}{ \partial s_t} \frac{ s_t}{ V_2}$$

Using (6) we compute  ${\cal E}_V(s_t).$  Thus, 
 $${\cal E}_V(s_t) =  \gamma \frac{ s_t}{l_t - s_t}$$
 Manipulating and using (3) we get

  $${\cal E}_V(s_t) =  \gamma \frac{  \hat h} { 1 - \hat h}$$
  Considering again the  balanced path  and using (28) we write  ${\cal E}_V(s_t)$ in terms of parameters 
 $${\cal E}_V(s_t)   = \gamma \frac{    \frac{(\gamma + \sigma)(\rho -1)}{1 - \sigma(2+\gamma)}    } {  1 -   \frac{(\gamma + \sigma)(\rho -1)}{1 - \sigma(2+\gamma)} } $$
 We know that the intertemporal elasticity  of substitution of time spent on   distracting activities $IES$  is defined as $\frac{1} { {\cal E}_V(s_t) }$ so that 
 $$IES (s) =  \frac{(1+\gamma) (1-\sigma) - \rho(\gamma + \sigma)}{ \gamma(\gamma + \sigma)(\rho-1)}\eqno(29)$$


For $\sigma > 1$  and $\rho < 1,$  we clearly have that $IES (s)$ tends to 0 as $\gamma \to \infty.$

\section{Analysis  of results and concluding remarks }

First, in relation to the home-office job supply represented by $l_t - s_t$  we know from  (6) that the intertemporal elasticity  of substitution of the home office  is $\frac{1}{\gamma}$ and   the intertemporal elasticity  of substitution of consumption is  $\frac{1}{\sigma}.$  Second, for the balanced   growth path satisfying  (8) - (12) to be a  solution of  the  benevolent planner's problem it  is sufficient  that the transversality conditions are satisfied. This  is achieved  by   assuming  $\sigma > 1.$   Third,  differently from  Eriksson's (1996) model we have considered the productivity of human capital sector as being 1 and the discount factor $\rho < 1.$  Lastly, if workers had been patients ( $\rho = 1$),  we would have considered   the productivity of the human capital sector   as being greater than 1 in order   to keep the results  similar to that of  Eriksson as is shown in Propositions 2, 3 and Corollary 1.

Using all the  results of the previous paragraph, we can then  say that  the intertemporal elasticity  of substitution  of time spent on distracting activities, $IES (s)$  is  small enough provided that the  intertemporal elasticity    of substitution of  home-office is  small enough. More precisely one has that 
$$\lim_{ \gamma \to \infty}  IES (s) = 0$$

The intuition behind this result is that if workers want to avoid fluctuations in home-office labor, they should display strong preference to avoid fluctuations on distracting activities. This result does not seem to be plausible in the short-run due to the high volatility of the distracting activities because of the COVID-19 pandemic that has spread  throughout the world provoking fear and anxiety in citizens and particularly in workers. However, to have $IES(s)$ small enough does seem to be quite plausible in the
 long-run since workers will end up incorporating home office work if it is adopted as a form of labor.

We finish this section by saying that although our paper is deterministic, it does explain to a certain degree, the long-run behaviour of the home-office job supply in terms of time spent on distracting activities. More precisely, a necessary condition for the home-office job supply to be smooth is that the intertemporal elasticity of substitution of distracting activities be small enough, as shown in the previous limit. We hope that in future research the home-office job supply will be analysed in ampler settings, including markets and government.

\end{document}